\documentclass[pra,aps,twocolumn,superscriptaddress,amsmath,amssymb]{revtex4-1}

\usepackage[T1]{fontenc}
\usepackage[utf8]{inputenc}
\usepackage{lmodern}
\usepackage{graphicx}
\usepackage{cancel}
\usepackage{color}

\usepackage{amsthm}
\newtheorem{theorem}{Theorem}
\newtheorem{cor}{Corollary}
\newtheorem{Lemma}{Lemma}

\begin{document}

\title{Entanglement robustness against particle loss in multiqubit systems}

\author{A. Neven}
\affiliation{Institut de Physique Nucl\'eaire, Atomique et de Spectroscopie, \\ CESAM, University of Li\`ege, B\^atiment~B15, Sart Tilman, Li\`ege 4000, Belgium}

\author{J. Martin}
\affiliation{Institut de Physique Nucl\'eaire, Atomique et de Spectroscopie, \\ CESAM, University of Li\`ege, B\^atiment~B15, Sart Tilman, Li\`ege 4000, Belgium}

\author{T. Bastin}
\affiliation{Institut de Physique Nucl\'eaire, Atomique et de Spectroscopie, \\ CESAM, University of Li\`ege, B\^atiment~B15, Sart Tilman, Li\`ege 4000, Belgium}

\date{January 7, 2019}

\begin{abstract}
When some of the parties of a multipartite entangled pure state are lost, the question arises whether the residual mixed state is also entangled, in which case the initial entangled pure state is said to be robust against particle loss. In this paper, we investigate this entanglement
robustness for $N$-qubit pure states. We identify exhaustively all entangled states that are fragile, i.e., not robust, with respect to the loss of any single qubit of the system. We also study the entanglement robustness properties of symmetric states and put these properties in the perspective of the classification of states with respect to stochastic local operations assisted with classic communication (SLOCC classification).
\end{abstract}

\maketitle

\section{Introduction}

Entanglement is a key feature of quantum mechanics, intimately linked with its non-local nature. The intensive research undertaken in the past decades to better characterize quantum entanglement revealed its theoretical complexity as well as its potential for innovative experimental protocols (see for example Ref.~\cite{NielsenChuang,HorodeckiReview} and references therein). The practical realization of theoretical protocols involving entangled states constitutes an interesting challenge as it requires to solve experimental issues such as decoherence. In this context, the robustness of entanglement was introduced by Vidal and Tarrach~\cite{VidalTarrach} as a measure to quantify how resilient an entangled state is in the presence of local noise (see, e.g., Ref.~\cite{Goy} for recent results).

Entanglement robustness can also be defined with respect to particle loss~\cite{DurVidalCirac}. In this context, an entangled multipartite state is said to be robust (or in contrast fragile) with respect to the loss of a given subset of its particles if the reduced state of the remaining system is entangled (respectively separable). The Greenberger-Horne-Zeilinger (GHZ) state of 3 qubits $|\mathrm{GHZ}_3\rangle \equiv (|000\rangle + |111\rangle)/\sqrt{2}$ is an emblematic example of fragile entangled state. Though highly entangled as a 3-qubit state, it loses all entanglement with the loss of any qubit~\cite{DurVidalCirac, Dur}. By contrast, the so-called W state of 3 qubits $|\mathrm{W}_3\rangle \equiv (|001\rangle + |010\rangle + |100\rangle)/\sqrt{3}$, which captures the other type of genuine 3-qubit entanglement~\cite{DurVidalCirac}, has the highest robustness of entanglement against particle loss among all 3-qubit states in the sense that the average entanglement of its 2-qubit reduced density operators reaches the highest possible value~\cite{DurVidalCirac}. This property of maximal bipartite entanglement (as measured by the mean concurrence of the 2-qubit reduced density operators) was generalized to $|\mathrm{W}\rangle$ states of an arbitrary number of qubits~\cite{Koashi}.

Entanglement robustness and fragility have been mainly studied in the context of the loss of all particles but two~\cite{Dur,Koashi,Stockton} where the concurrence~\cite{Wootters} allows for a direct detection of the residual entanglement of the resulting 2-qubit mixed state. More general cases in which a finite fraction of particles is lost has been considered in~\cite{Mishra}. In Ref.~\cite{Helwig}, the concept of absolutely maximally entangled states in qudit systems is introduced with states that exhibit a specific fragility property with respect to the loss of half or more of the system. In Ref.~\cite{Rajagopal}, the relationship between robustness against particle loss and permutation symmetry was studied for specific 3-qubit pure states. The robustness of entanglement was also studied with respect to local measurements~\cite{Briegel}. 

An exhaustive study of all multiqubit fragile states with respect to the loss of any number of qubits is still missing. In this paper such a study is provided for general $N$-qubit systems when the loss of any single of the qubits is considered. We further investigate the influence of permutation invariance on entanglement fragility. We show that symmetric entangled states that are fragile with respect to the loss of one qubit are all stochastically equivalent through local operations assisted with classical communication (SLOCC-equivalent~\cite{Bennett}) and belong to the SLOCC class of the GHZ state. The robustness of entanglement for all states that belong to the SLOCC classes of the entangled Dicke states~\cite{footnote1} is also investigated.

The paper is organized as follows. In Sec.~\ref{1qubit}, we recall the definition of entanglement robustness and fragility against particle loss and focus on the scenario in which a single qubit is lost. In Sec.~\ref{symSubspace}, we investigate the consequences of permutational invariance on entanglement fragility. We draw a conclusion in Sec.~\ref{conclusion}.

\section{Fragility with respect to the loss of a single qubit}
\label{1qubit}

In an $N$-particle system, an entangled state $|\psi\rangle$ is said to be fragile (robust) with respect to the loss of a given subset $\mathcal{S}$ of the particles if $\rho_{\neg \mathcal{S}}(\psi) \equiv \mathrm{Tr}_{\mathcal{S}}(|\psi\rangle\langle\psi|)$ is separable (entangled), where $\mathrm{Tr}_{\mathcal{S}}(|\psi\rangle\langle\psi|)$ is the partial trace of the density operator $|\psi\rangle\langle\psi|$ over the subset $\mathcal{S}$. In particular, an entangled state $|\psi\rangle$ is fragile (robust) with respect to the loss of the $k$th particle ($1 \leqslant k \leqslant N$) if
\begin{equation}
\rho_{\neg k}(\psi) \equiv \mathrm{Tr}_k(|\psi\rangle\langle\psi|)
\end{equation}
is separable (entangled), where $\mathrm{Tr}_k(|\psi\rangle\langle\psi|)$ is the partial trace of the density operator $|\psi\rangle\langle\psi|$ over the $k$th particle.
The fragility (robustness) of multiparticle entangled states is an invariant under local unitary (LU) operations. Indeed, for any subset $\mathcal{S}$ of the particles, any local unitaries acting on an entangled state $|\psi\rangle$ results in a similar LU action on the remaining particle reduced state $\rho_{\neg \mathcal{S}}(\psi)$ and the separability or entanglement of a mixed state is an LU-invariant~\cite{HorodeckiReview}.

The concept of entanglement robustness and fragility against particle loss makes only sense for $N$-partite systems of at least three parties, so that we always assume hereafter $N \geqslant 3$. In addition, throughout the paper, equality of state vectors is meant up to a global phase factor.

We first prove a simple lemma showing the incompatibility between the fragility of an entangled $N$-qubit state with respect to the loss of the $k$th qubit and biseparability for the bipartition $k|1\cdots \cancel{k} \cdots N$.
\begin{Lemma}
If $|\psi\rangle$ is a pure $N$-qubit entangled state that is fragile with respect to the loss of the $k$th qubit, then $\mathrm{rank}[\rho_{\neg k}(\psi)] = 2$. \label{lemmaRank}
\end{Lemma}
\begin{proof}
As a consequence of the Schmidt decomposition~\cite{NielsenChuang} of the pure state $|\psi\rangle$ for the bipartition $k|1\cdots \cancel{k} \cdots N$, the reduced density operator $\rho_{\neg k}(\psi)$ has a rank that is at most equal to 2. A reduced density operator of rank 1 corresponds to the case where $|\psi\rangle$ is biseparable for the considered bipartition~\cite{HorodeckiReview}. In this case, the reduced density operator is a pure state, that is separable because of the fragility hypothesis of the state $|\psi\rangle$. This leads however to a fully separable state $|\psi\rangle$, which contradicts our hypothesis about the entanglement of $|\psi\rangle$ and concludes the proof of the lemma.
\end{proof}

We can now state our main result.
\begin{theorem}
An entangled $N$-qubit state $|\psi\rangle$ is fragile with respect to the loss of any qubit that is part of a given nonempty subset $\mathcal{A}$ of the qubits if and only if it can be written in the form
\begin{equation}
|\psi\rangle = \sqrt{p} \, |e_1,\dots,e_N\rangle + \sqrt{1-p} \, |e_1^\prime,\dots,e_N^\prime\rangle,
\label{canonicForm}
\end{equation}
where $0 < p < 1$, $|e_i\rangle, \; |e_i^\prime\rangle$ ($i=1,\dots,N$) are normalized single-qubit states with $|e_i^\prime\rangle \perp |e_i\rangle, \forall i \in \mathcal{A}$ and, if $\#\mathcal{A}=1$, $|e_j^\prime\rangle \neq |e_j \rangle$ for at least one qubit $j \notin \mathcal{A}$.
\label{mainThm}
\end{theorem}

\begin{proof}
An entangled state $|\psi\rangle$ fragile with respect to the loss of any particle belonging to a given subset $\mathcal{A}$ of the particles is a state such that $\rho_{\neg k}(\psi)$ is separable, $\forall k \in \mathcal{A}$. We first show that the expressed condition is sufficient: for a given nonempty subset $\mathcal{A}$ of an $N$-qubit system, any state of the form (\ref{canonicForm}) is entangled and fragile with respect to the loss of any qubit that belongs to $\mathcal{A}$. Indeed, for such states and any $k \in \mathcal{A}$, $\langle e_k^\prime|e_k\rangle = 0$ and $\rho_{\neg k}(\psi) = p |\mathbf{e}_{\neg k}\rangle \langle \mathbf{e}_{\neg k} | + (1 - p) |\mathbf{e}_{\neg k}^\prime\rangle \langle \mathbf{e}_{\neg k}^\prime |$, with $|\mathbf{e}_{\neg k}\rangle = \otimes_{i \neq k} |e_i\rangle$ and $|\mathbf{e}'_{\neg k}\rangle = \otimes_{i \neq k} |e'_i\rangle$, hence $\rho_{\neg k}(\psi)$ is separable. In addition, $\rho_{\neg k}(\psi)$ is a convex sum of two projectors onto linearly independent vectors and is thus of rank 2. This ensures that the state $|\psi\rangle$ is itself entangled~\cite{HorodeckiReview}.

We now show that the condition is necessary. Let $|\psi\rangle$ be an entangled $N$-qubit state that is fragile with respect to the loss of any qubit belonging to a given subset $\mathcal{A}$ that is supposed to contain at least one qubit, say the $k$th. According to Lemma~\ref{lemmaRank}, the separable reduced density operator $\rho_{\neg k}(\psi)$ is of rank 2. It can thus be written as a convex sum of projectors onto two distinct $(N-1)$-qubit product states $|\mathbf{e}_{\neg k}\rangle \equiv \otimes_{i \neq k} |e_i\rangle$ and $|\mathbf{e}'_{\neg k}\rangle \equiv \otimes_{i \neq k} |e'_i\rangle$ (with $|e_i\rangle$ and $|e^\prime_i\rangle \; (i \neq k)$ normalized states of the $i$th qubit not all pairwise identical)~\cite{Chen2013}:
\begin{equation}
\rho_{\neg k}(\psi) = p \: |\mathbf{e}_{\neg k}\rangle \langle \mathbf{e}_{\neg k} | + (1 - p) \: |\mathbf{e}_{\neg k}^\prime\rangle \langle \mathbf{e}_{\neg k}^\prime |, \label{partialTrace1}
\end{equation}
with $0 < p < 1$. In addition, the spectral decomposition of $\rho_{\neg k}(\psi)$ is only composed of two terms and we have
\begin{equation}
\rho_{\neg k}(\psi) = \lambda_1 \: |v_1\rangle \langle v_1 | + \lambda_2 \: |v_2\rangle \langle v_2 |, \label{eigendecomposition}
\end{equation}
with $\lambda_1$ and $\lambda_2$ the two nonzero eigenvalues of $\rho_{\neg k}(\psi)$ with associated eigenvectors $|v_1\rangle$ and $|v_2\rangle$, respectively. A $2 \times 2$ unitary $U$ can always be found to connect both decompositions (\ref{partialTrace1}) and (\ref{eigendecomposition}) such that~\cite{Hughston}
\begin{equation}
\left( \begin{array}{c}
\sqrt{\lambda_1} \: |v_1\rangle \\
\sqrt{\lambda_2} \: |v_2\rangle
\end{array} \right) = U \left( \begin{array}{c}
\sqrt{p} \: |\mathbf{e}_{\neg k}\rangle \\
\sqrt{1-p} \: |\mathbf{e}_{\neg k}^\prime\rangle
\end{array} \right).
\label{unitary}
\end{equation}
Considering the Schmidt decomposition~\cite{NielsenChuang} of the state $|\psi\rangle$ for the bipartition $k|1\cdots \cancel{k} \cdots N$, two orthonormal single qubit states $|a_1\rangle,|a_2\rangle$ can be found such that
\begin{equation}
|\psi\rangle = \sqrt{\lambda_1} \; |v_1\rangle \otimes |a_1\rangle_k + \sqrt{\lambda_2} \; |v_2\rangle \otimes |a_2\rangle_k,
\end{equation}
where the subscript $k$ explicitly specifies the single qubit the states $|a_1\rangle$ and $|a_2\rangle$ refer to. It then follows from Eq.~\eqref{unitary} that
\begin{equation}
|\psi\rangle = \sqrt{p} \: |\mathbf{e}_{\neg k}\rangle \otimes |e_k\rangle \; + \; \sqrt{1-p} \: |\mathbf{e}_{\neg k}^\prime\rangle \otimes |e_k^\prime\rangle
\label{psifromTr1},
\end{equation}
where
\begin{equation}
\left( \begin{array}{c}
|e_k\rangle \\
|e_k^\prime\rangle
\end{array} \right) = U^T \left( \begin{array}{c}
|a_1\rangle \\
|a_2\rangle
\end{array} \right).
\label{UnitaryT}
\end{equation}
Since $|a_1\rangle$ and $|a_2\rangle$ are orthonormal, so are $|e_k\rangle$ and $|e'_k\rangle$. In the case where the subset $\mathcal{A}$ only contains the qubit $k$, Eq.~(\ref{psifromTr1}) is nothing but Eq.~(\ref{canonicForm}) where $|e'_k\rangle \perp |e_k\rangle$ and at least one state $|e'_i\rangle$ is not equal to $|e_i\rangle$ for all $i\neq k$. This ends the proof of the theorem in this case.

If additional qubits belong to the subset $\mathcal{A}$, let for instance the $l$th ($l \neq k$) be part of it. This implies that the reduced density operator $\rho_{\neg l}(\psi)$ is separable. Considering Eq.~(\ref{psifromTr1}), we explicitly have
\begin{align}
\rho_{\neg l}(\psi) & = p |\mathbf{e}_{\neg l}\rangle \langle\mathbf{e}_{\neg l}| + (1-p) |\mathbf{e}_{\neg l}^\prime\rangle \langle\mathbf{e}_{\neg l}^\prime| \nonumber \\
& + \sqrt{p(1-p)} (\langle e_l | e_l^\prime\rangle |\mathbf{e}_{\neg l}^\prime\rangle\langle\mathbf{e}_{\neg l}| + \mathrm{H.c.}),
\label{rhonegl}
\end{align}
with $|\mathbf{e}_{\neg l}\rangle = \otimes_{i\neq l} |e_i\rangle$, $|\mathbf{e}_{\neg l}^\prime\rangle = \otimes_{i\neq l} |e_i^\prime\rangle$ and where $\mathrm{H.c.}$ stands for Hermitian conjugate.
In addition, since for at least one index $i \neq k$, $|e_i^\prime\rangle \neq |e_i\rangle$, either (a) $|\mathbf{e}_{\neg k,l}^\prime\rangle = |\mathbf{e}_{\neg k,l}\rangle$ or (b) $|\mathbf{e}_{\neg k,l}^\prime\rangle \neq |\mathbf{e}_{\neg k,l}\rangle$, where $|\mathbf{e}_{\neg k,l}\rangle \equiv \otimes_{i\neq k,l}|e_i\rangle$ and $|\mathbf{e}_{\neg k,l}^\prime\rangle \equiv \otimes_{i\neq k,l}|e_i^\prime\rangle$. In case (a), $|\psi\rangle = |\mathbf{e}_{\neg k,l}\rangle \otimes |u_{k,l}\rangle$, with $|u_{k,l}\rangle = \sqrt{p} |e_k,e_l\rangle + \sqrt{1-p}|e_k^\prime,e'_l\rangle$, and $\rho_{\neg l}(\psi)$ simplifies to $|\mathbf{e}_{\neg k,l}\rangle \langle\mathbf{e}_{\neg k,l}| \otimes \rho_k$ with $\rho_k$ the $k$th qubit state $\rho_k = p |e_k\rangle\langle e_k| + (1-p)|e'_k\rangle\langle e'_k| + \sqrt{p(1-p)}(\langle e_l|e'_l\rangle |e'_k\rangle\langle e_k| +\mathrm{H.c.})$.
The reduced density matrix $\rho_{\neg l}(\psi)$ is separable whatever the $l$th qubit states $|e_l\rangle$ and $|e'_l\rangle$. However, since by hypothesis $|\psi\rangle$ is entangled, so is the state $|u_{k,l}\rangle$. Hence, its Schmidt rank is 2 and it admits a Schmidt decomposition $|u_{k,l}\rangle = \sqrt{p'} |f_k,f_l\rangle + \sqrt{1-p'}|f_k^\perp,f_l^\perp\rangle$, with $0 < p' < 1$ and ($|f_k\rangle$, $|f_k^\perp \rangle$), ($|f_l\rangle$, $|f_l^\perp \rangle$) 2 pairs of orthonormal states~\cite{NielsenChuang}. In this form, the state $|\psi\rangle$ is a linear superposition of two product states with orthonormal states for qubits $k$ and $l$.
In case (b), there exists a state $|\mathbf{e}_{\neg k,l}^\perp\rangle \perp |\mathbf{e}_{\neg k,l}\rangle$ that has a nonzero overlap with $|\mathbf{e}_{\neg k,l}^\prime\rangle$. In this case, the separability of $\rho_{\neg l}(\psi)$ is equivalent to the condition $|e'_l\rangle \perp |e_l\rangle$. Indeed, on the one hand, if $|e'_l\rangle \perp |e_l\rangle$, $\rho_{\neg l}(\psi)$ is trivially separable [see Eq.~(\ref{rhonegl})]. On the other hand, if $|e'_l\rangle$ is not orthogonal to $|e_l\rangle$, $\langle \mathbf{e}_{\neg k,l}^\perp, e'_k|\rho_{\neg l}(\psi)| \mathbf{e}_{\neg k,l}, e_k\rangle \neq 0$ and since $\langle \mathbf{e}_{\neg k,l}^\perp, e_k|\rho_{\neg l}(\psi)| \mathbf{e}_{\neg k,l}^\perp, e_k\rangle = 0$ this implies that $\rho_{\neg l}(\psi)$ is entangled (Lemma I of Ref.~\cite{Boyer}). Whatever the case (a) or (b), all this shows that if the state $|\psi\rangle$ is fragile with respect to the loss of qubits $k$ or $l$, it can be written in the form of Eq.~(\ref{canonicForm}) with $|e'_i\rangle \perp |e_i\rangle$ for $i = k,l$. If the subset $\mathcal{A}$ only contains these qubits, this ends the proof of the theorem.

For any additional qubit $l' \neq k, l$ that belongs to $\mathcal{A}$, $|\mathbf{e}_{\neg k,l'}^\prime\rangle$ necessarily differs from $|\mathbf{e}_{\neg k,l'}\rangle$ since $|e'_l\rangle \perp |e_l\rangle$ and the separability of $\rho_{\neg l'}(\psi)$ is equivalent to $|e'_{l'}\rangle \perp |e_{l'}\rangle$ since the case falls within case (b) here above that holds whatever $l \neq k$. This ends the proof of the theorem in the most general case.
\end{proof}

A straightforward corollary of Theorem~\ref{mainThm} can be stated as follows.
\begin{cor}
An entangled $N$-qubit state $|\psi\rangle$ is fragile with respect to the loss of any qubit if and only if it can be written in the form
\begin{equation}
|\psi\rangle = a |e_1,\dots,e_N\rangle + b |e_1^\perp,\dots,e_N^\perp\rangle, \label{RegPol}
\end{equation}
where $(|e_i\rangle,|e_i^\perp\rangle)$ $(i=1,\dots,N)$ are pairs of orthonormal single-qubit states and $a,b$ are two nonzero complex numbers with $|a|^2 + |b|^2 = 1$. The phases of $|e_1\rangle$ and $|e_1^\perp\rangle$ can always be chosen to have both $a$ and $b$ real and positive. \label{corollaryNqubits}
\end{cor}

An emblematic example of state that is fragile with respect to the loss of any qubit is given by the state $|\mathrm{GHZ}_N\rangle \equiv ( |0,\ldots,0\rangle + |1,\ldots,1\rangle )/\sqrt{2}$~\cite{DurVidalCirac}. Actually, all $N$-qubit states that are fragile with respect to the loss of any qubit are SLOCC-equivalent to this $|\mathrm{GHZ}_N\rangle$ state and hence belong to the same SLOCC class. Indeed, such states $|\psi\rangle$ being of the form (\ref{RegPol}) with $(|e_i\rangle,|e_i^\perp\rangle)$ pairs of orthonormal states $\forall i$, there necessarily exist unitary matrices $U_i$ with $U_i |e_i\rangle = |0\rangle$ and $ U_i |e_i^\perp\rangle = |1\rangle, \forall i$. This implies
\begin{equation}
|\mathrm{GHZ}_N\rangle = \left( \bigotimes_{i=1}^{N} D U_i \right) |\psi\rangle,
\label{ILO}
\end{equation}
with $D = \mathrm{diag}[(\sqrt{2} a)^{-1/N},(\sqrt{2} b)^{-1/N}]$. Since $a$ and $b$ are both nonzero, $D$ is invertible and the states $|\psi\rangle$ and $|\mathrm{GHZ}_N\rangle$ are connected through an invertible local operation (ILO), hence they are SLOCC-equivalent~\cite{DurVidalCirac}. This does not imply that the SLOCC class of the GHZ state would only be composed of fragile states because the ILOs of Eq.~(\ref{ILO}) do only form a subset of all possible ILOs that are defined by \emph{arbitrary} invertible matrices. We can thus state the following theorem.
\begin{theorem}
The entangled $N$-qubit states that are fragile with respect to the loss of any qubit all belong to the same SLOCC class, namely that of the $|\mathrm{GHZ}_N\rangle$ state (which still also contains robust states). \label{SLOCCequivalenceCor}
\end{theorem}

\section{Robustness against particle loss in the symmetric subspace}
\label{symSubspace}

A symmetric $N$-qubit state is a state $|\psi_S\rangle$ that is invariant under any permutation of the qubits: $ \hat{P}_\pi |\psi_S\rangle = |\psi_S\rangle \; \forall \; \pi \in S_N$, where $S_N$ is the permutation group of $N$ elements and $\hat{P}_\pi$ is the permutation operator corresponding to the permutation $\pi$. Because of this invariance, the reduced density operators of symmetric multiqubit states are all equal when $t < N$ qubits are traced out whatever these $t$ qubits. This implies in particular that fragility or robustness with respect to the loss of $t$ given qubits holds whatever the $t$ qubits selected. In view of this observation, we first particularize the results of the previous section about the fragility of multiqubit states with respect to the loss of one qubit to the special case of symmetric states. We then investigate the fragility properties of symmetric states with respect to the loss of more than one qubit. Finally, the entanglement robustness of the states SLOCC-equivalent to the entangled symmetric Dicke states is studied.

\subsection{Fragility with respect to the loss of a single qubit}

Symmetric states that are fragile with respect to the loss of any single qubit satisfy Corollary~\ref{corollaryNqubits}. In the symmetric subspace, this corollary can be refined as follows.
\begin{cor}
A symmetric entangled $N$-qubit state $|\psi_S\rangle$ is fragile with respect to the loss of any qubit if and only if it can be written in the form
\begin{equation}
|\psi_S\rangle = a |e,\dots,e\rangle + b |e^\perp,\dots,e^\perp\rangle ,
\label{RegPolSym}
\end{equation}
where $|e\rangle$ and $|e^\perp\rangle$ are two orthonormal single-qubit states and $a,b$ are 2 nonzero complex numbers with $|a|^2 + |b|^2 = 1$. The phases of $|e\rangle$ and $|e^\perp\rangle$ can always be chosen to have both $a$ and $b$ real and positive.
\label{SymCor}
\end{cor}

\begin{proof}
Of course, the states~\eqref{RegPolSym} are symmetric and constitute a subset of states~\eqref{RegPol} that are fragile with respect to the loss of any qubit. Conversely, the only states $|\psi\rangle$ of the form of Eq.~\eqref{RegPol} that are also symmetric are necessarily of the form of Eq.~\eqref{RegPolSym}. Indeed, such symmetric states $|\psi\rangle$ satisfy the equalities $\langle \mathbf{e}_{\neg k,l},e_k,e_l| P_{kl} |\psi\rangle = \langle \mathbf{e}_{\neg k,l},e_k,e_l|\psi\rangle$, $\forall k,l = 1, \ldots, N$, with $k \neq l$, where $P_{kl}$ is the permutation operator between qubits $k$ and $l$. These equalities are equivalent to $|e_k\rangle = |e_l\rangle$, $\forall k,l: k \neq l$. This implies, up to a global phase factor in each single qubit state, $|e_1\rangle = \ldots = |e_N\rangle$ and $|e_1^\perp\rangle = \ldots = |e_N^\perp\rangle$. Setting $|e\rangle = |e_1\rangle$, $|e^\perp\rangle = |e_1^\perp\rangle$, and factoring out each phase factor of the individual qubit states into the coefficients $a$ and $b$ casts $|\psi\rangle$ into the form (\ref{RegPolSym}).
\end{proof}

A geometrical picture linked to the fragility of symmetric states with respect to the loss of any qubit can also be given. In the Majorana representation~\cite{Majorana}, an $N$-qubit symmetric state is univocally associated to a set of $N$ single-qubit states and it can be geometrically represented by the corresponding $N$ points on the Bloch sphere, the so-called Majorana points~\cite{Martin2010}. The set of all symmetric states consists of all possible configurations of the $N$ Majorana points, some of which can be superimposed. For instance, the Majorana representation of the $|\mathrm{GHZ}_N\rangle$ state is a set of $N$ points at the vertices of a regular $N$-sided polygon in the equatorial plane of the Bloch sphere~\cite{Martin2010}. This regular polygonal arrangement of the Majorana points on the Bloch sphere is actually the very general signature of the fragility of a symmetric state with respect to the loss of any qubit:
\begin{theorem}
A symmetric entangled $N$-qubit pure state $|\psi_S\rangle$ is fragile with respect to the loss of any qubit if and only if its Majorana points are at the vertices of a regular $N$-sided polygon in whichever plane intersecting the Bloch sphere.
\end{theorem}

\begin{proof}
The point-group symmetries of the Majorana points of a multiqubit symmetric state induce constraints on the coefficients of the symmetric state in the Dicke basis~\cite{Baguette2015}. In particular, an $N$-qubit symmetric state has $N$ distinct Majorana points exhibiting the symmetries of a regular $N$-sided polygon if and only if it is equivalent through a symmetric LU operation to a state of the form
\begin{equation}
|\phi_S\rangle = a^\prime |0,\dots,0\rangle + b^\prime |1,\dots,1\rangle,
\label{prop1proof}
\end{equation}
where $a^\prime$ and $b^\prime$ are two nonzero complex numbers with $|a^\prime|^2 + |b^\prime|^2 = 1$~\cite{Baguette2015}. As the single-qubit states $|0\rangle$ and $|1\rangle$ are orthonormal, a unitary operation always transforms them into a pair of orthonormal states. As a consequence, a symmetric LU operation always transforms a state of the form~\eqref{prop1proof} into a state of the form~\eqref{RegPolSym} and Corollary~\ref{SymCor} concludes the proof.
\end{proof}

\subsection{Fragility with respect to the loss of several qubits}

Theorem 2 shows that fragile symmetric states with respect to the loss of any single qubit all belong to the same SLOCC class. This is, however, not true when considering fragility with respect to the loss of at least two qubits. We illustrate this with symmetric states of four qubits. Symmetric multiqubit states can be classified into entanglement families based on their Majorana representation~\cite{Bastin_families}. In this classification, symmetric 4-qubit states with four distinct points on the Bloch sphere belong to the family $\mathcal{D}_{1,1,1,1}$, which contains an infinite number of different SLOCC classes~\cite{Bastin_families}. Each SLOCC class of this family is unambiguously represented by the state~\cite{Baguette}
\begin{equation}
|\psi_\mu \rangle = \frac{1}{\sqrt{2 + |\mu|^2}} \Big( |D_4^{(0)} \rangle + \mu |D_4^{(2)} \rangle + |D_4^{(4)} \rangle \Big),
\end{equation}
with $\mu$ a complex number that belongs to the set
\begin{multline}
S=\{ \mu \in \mathbb{C}: \mathrm{Re}(\mu) \geqslant 0, \; \mu < \sqrt{2/3} \textrm{ if } \mathrm{Im}(\mu) = 0 , \\ \mathrm{Im}(\mu) \geqslant 0 \textrm{ and } |\mu - \sqrt{2/3}| < \sqrt{8/3} \textrm{ if } \mu \neq \sqrt{2} \; i\},
\end{multline}
and where $|D_N^{(k)}\rangle$ stands for the $k$-excitation Dicke state~\cite{footnote1}. In their respective SLOCC class, the states $|\psi_\mu \rangle$ are up to LU operations the unique maximally entangled states in the sense that their 1-qubit reduced density operators are maximally mixed~\cite{Baguette}. Their robustness with respect to the loss of one or two qubits can be estimated using an entanglement measure for the residual mixed state. To this aim, the negativity is a good tool. For a mixed state $\rho$, it is defined as~\cite{VidalWerner} $\mathcal{N}(\rho)= (|| \rho^{T_{\mathcal{A}}} ||_1 - 1)/2$, where $|| \rho ||_1=\mathrm{Tr}(\sqrt{\rho \rho^\dagger})$ and $\rho^{T_\mathcal{A}}$ is the partial transpose of $\rho$ with respect to a strict subset $\mathcal{A}$ of the parties. After the loss of $t=1$ or 2 qubits, the separability of the residual 3- or 2-qubit symmetric mixed state $\rho$ is equivalent to $\mathcal{N}(\rho) = 0$~\cite{Augusiak}. For $t=1$, only the state $|\psi_{\mu = 0}\rangle = |\mathrm{GHZ}_4\rangle$ is fragile. For $t=2$, we compute that the negativity of the residual mixed state vanishes on the continuous subset
\begin{equation}
    S \cap \left\{ \mu \in \mathbb{C}: \mathrm{Im}(\mu) \geqslant \sqrt{\left[\sqrt{6}-\mathrm{Re}(\mu)\right] \mathrm{Re}(\mu)} \right\}
\end{equation}
Since the states $|\psi_\mu \rangle$ are SLOCC-inequivalent for different values of $\mu$ in $S$, this shows that the fragility of multiqubit symmetric entangled states for the loss of more than one qubit is not restricted to a single SLOCC class (in contrast to the fragility with respect to the loss of a single qubit).

\subsection{Robustness in the symmetric Dicke state SLOCC classes}

For any $N \geqslant 3$, the 2-qubit reduced density operator of all entangled symmetric Dicke states $|D_N^{(k)}\rangle$ ($k = 1, \ldots, N-1$) are themselves entangled~\cite{WangMolmer}. Since separable states remain obviously separable after particle loss, this implies that all entangled Dicke states $|D_N^{(k)}\rangle$ are robust with respect to the loss of any number $t$ of qubits ($t \leqslant N-2$). We conjecture that this robustness property holds for all symmetric states in the Dicke state SLOCC classes $\mathcal{D}_{N-k,k}$ ($k = 1, \ldots, \lfloor N/2 \rfloor$)~\cite{footnote1}. Indeed, in the Majorana representation the symmetric states that belong to any of these SLOCC classes are of the form~\cite{Bastin_families} $\mathcal{N}_S \sum_\pi |\epsilon, \dots, \epsilon, \epsilon', \dots, \epsilon'\rangle$, where $\mathcal{N}_S$ is a normalization constant, $|\epsilon\rangle$ and $|\epsilon'\rangle$ are two distinct normalized single qubit states, the multiqubit states in the sum contain $N-k$ qubits in state $|\epsilon\rangle$ and $k$ qubits in state $|\epsilon'\rangle$, and $\pi$ denotes all permutations of the qubits leading to different terms in the sum. Using LU operations it is always possible to map $|\epsilon\rangle$ and $|\epsilon'\rangle$ onto $|0\rangle$ and $(u|0\rangle + |1\rangle)/\sqrt{1+u^2}$, respectively, with $u \geqslant 0$. As a consequence, all symmetric states in the $\mathcal{D}_{N-k,k}$ SLOCC classes ($k > 0$) are LU-equivalent to a state of the form
\begin{equation}
|\psi_N^{(k)}(u) \rangle = \sqrt{A_N^{(k)}(u)} \sum_\pi |0\rangle^{\otimes N-k} \otimes ( u \: |0\rangle + |1\rangle )^{\otimes k} ,
\end{equation}
with
\begin{equation}
A_N^{(k)}(u) = \left[\sum_{i=0}^k{ \frac{{k \choose i}^2}{{N \choose i}} u^{2(k-i)} } \right]^{-1}.
\end{equation}

If we trace out all qubits but the first two, the reduced density operator $\rho_{1,2}[\psi_N^{(k)}(u)] \equiv \mathrm{Tr}_{\neg 1,2}\left[|\psi_N^{(k)}(u)\rangle \langle \psi_N^{(k)}(u)|\right]$ reads in the computational basis
\begin{multline}
\rho_{1,2} [\psi_N^{(k)}(u)] = A_N^{(k)}(u) \\
\left( \begin{array}{cccc}
f_N^{(k)}(u,0,0) & f_N^{(k)}(u,1,0) & f_N^{(k)}(u,1,0) & f_N^{(k)}(u,2,0) \\
f_N^{(k)}(u,1,0) & f_N^{(k)}(u,1,1) & f_N^{(k)}(u,1,1) & f_N^{(k)}(u,2,1) \\
f_N^{(k)}(u,1,0) & f_N^{(k)}(u,1,1) & f_N^{(k)}(u,1,1) & f_N^{(k)}(u,2,1) \\
f_N^{(k)}(u,2,0) & f_N^{(k)}(u,2,1) & f_N^{(k)}(u,2,1) & f_N^{(k)}(u,2,2)
\end{array} \right),
\end{multline}
with
\begin{equation}
f_N^{(k)}(u,j,j^\prime) = \frac{1}{ {\scriptstyle {N \choose k}^2}} \sum_{i=0}^k { {\scriptstyle {N-i-j \choose k-i-j} {N-i-j^\prime \choose k-i-j^\prime} {N-2 \choose i} } \; u^{2(k-i)-j-j^\prime} }.
\end{equation}

The robustness with respect to the loss of any number $t$ of qubits ($t \leqslant N-2$) of $|\psi_N^{(k)}(u)\rangle$ and their LU-equivalents is asserted if $\rho_{1,2}[\psi_N^{(k)}(u)]$ is entangled, which happens if the determinant of the partial transpose of $\rho_{1,2} [\psi_N^{(k)}(u)]$ is strictly negative. For $k=1$, this is indeed the case:
\begin{equation}
\det \left\{ \rho_{1,2} [\psi_N^{(1)}(u)]^{T_1} \right\} = - \left[ A_N^{(1)}(u) \right]^4 < 0, \,\,\, \forall u \geqslant 0.
\end{equation}
Hence, all $\mathcal{D}_{N-1,1}$ SLOCC class states are robust with respect to the loss of any number $t$ of qubits ($t \leqslant N-2$).

For $k > 1$, only strong numerical evidences for a similar conclusion can be given. In Fig.~\ref{fixedN}, we see that for a fixed number $N$ of qubits (here set to 12) the determinant of the partial transpose of $\rho_{1,2}[\psi_N^{(k)}(u)]$ gets more and more negative for increasing values of $k$. The behavior for increasing number of qubits and $k = 2$ is shown in Fig.~\ref{fixedk}. In this case, the determinant flattens with increasing number of qubits, coming from the negative side closer to the flat zero curve. This suggests that extrapolating to any value of $N$ and $k$, we would always have strictly negative values for the determinant and that as a consequence all symmetric states in the remaining Dicke state SLOCC classes $\mathcal{D}_{N-k,k}$ ($k = 2, \ldots, \lfloor N/2 \rfloor$) would be robust with respect to the loss of any number $t$ of qubits ($t \leqslant N-2$).

\begin{figure}
\centering
\includegraphics[scale=0.49]{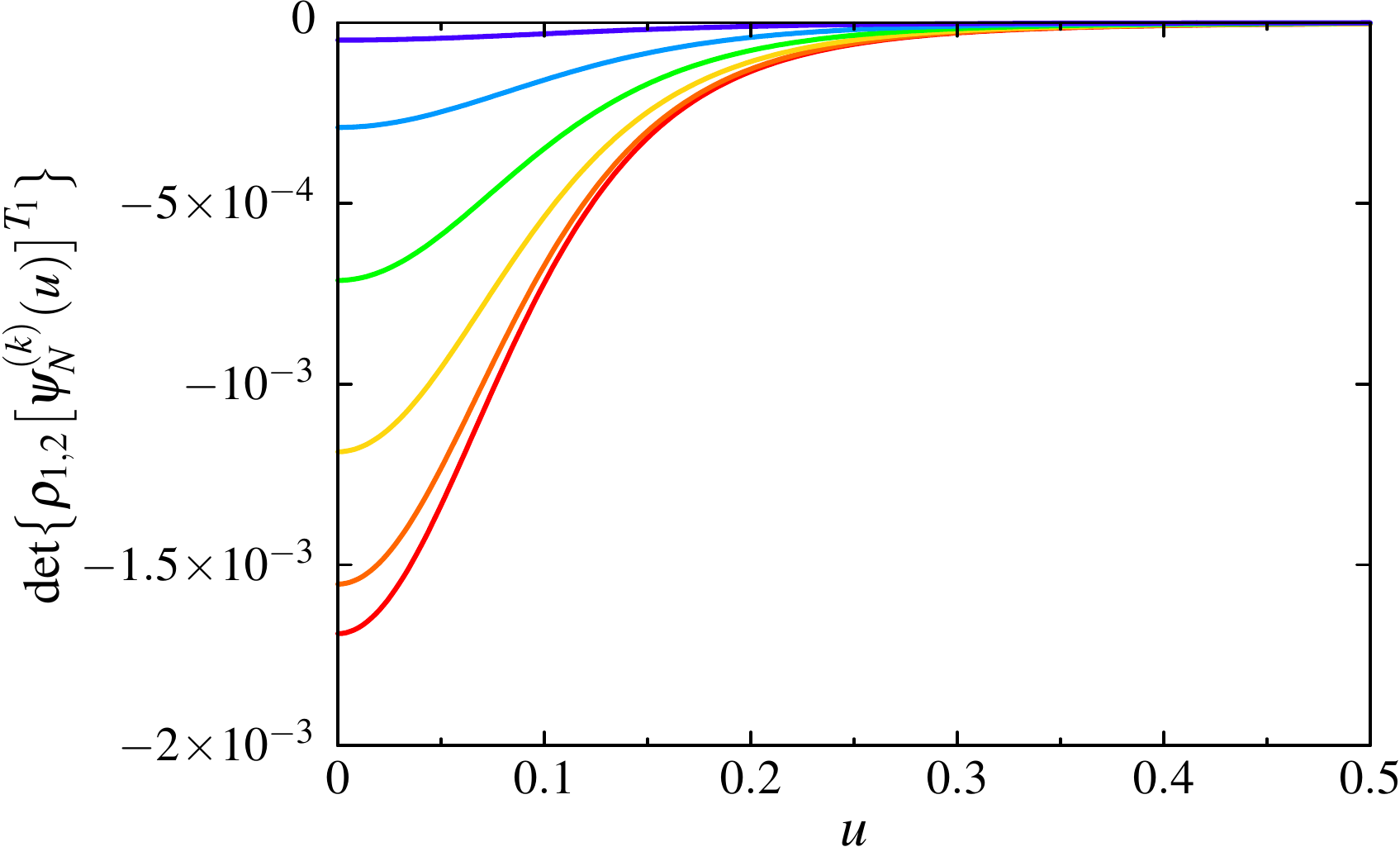}
\caption{Plot of the determinant of the partial transpose of $\rho_{1,2}[\psi_{N}^{(k)}(u)]$ as a function of $u$ for $N=12$ and (from top to bottom) $k=1,2,\ldots,6$.}
\label{fixedN}
\end{figure}

\begin{figure}
\centering
\includegraphics[scale=0.49]{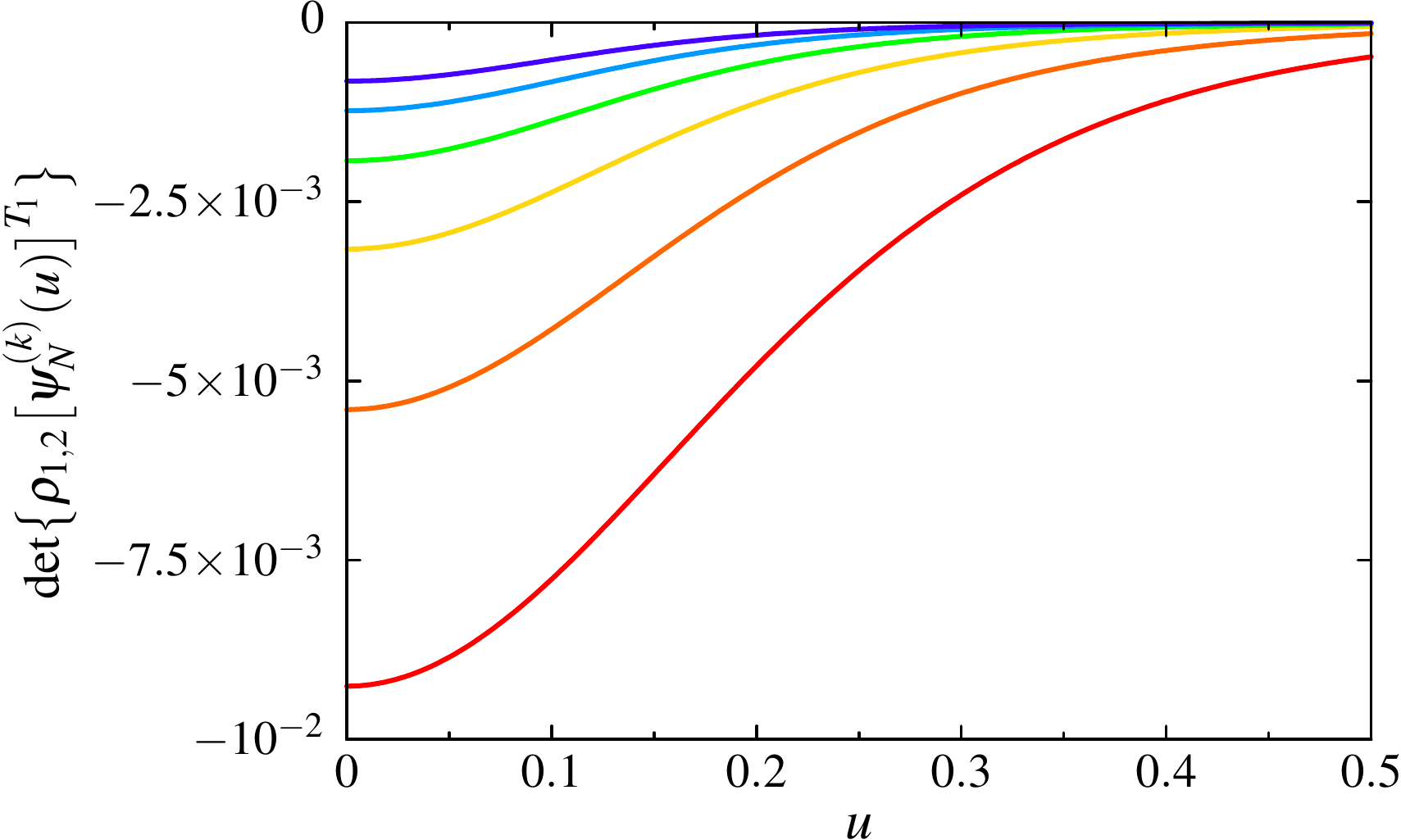}
\caption{Plot of the determinant of the partial transpose of $\rho_{1,2}[\psi_{N}^{(k)}(u)]$ as a function of $u$ for $k=2$ and (from top to bottom) $N=9,8,\ldots,4$.}
\label{fixedk}
\end{figure}

Identifying states that are robust with respect to the loss of an arbitrary number of qubits is particularly interesting as this property is not generic. When translated into the context and notations of this paper, the results of Refs.~\cite{Kendon,KendonZyczkowski,Aubrun} imply that $N$-qubit pure states chosen uniformly at random (according to the natural Fubini-Study measure in the $N$-qubit Hilbert space) are fragile with respect to the loss of any $t$ qubits, with a probability close to 1 (0) for the largest (smallest) values of $t$. There is a transition region between generic robustness and generic fragility, corresponding to the values of $t$ for which both fragile and robust states with respect to the loss of $t$ qubits can be found with nonnegligible probability. In this region of the $t$ parameter space, the probability of being robust was shown to decrease exponentially with $t$ from both numerical simulations~\cite{Kendon,KendonZyczkowski} and theoretical arguments~\cite{Aubrun}, making the transition surprisingly sharp. For $N$-qubit systems, the transition was numerically estimated to lie somewhere between $N/3$ and $N/2$~\cite{KendonZyczkowski}.

\section{Conclusion}
\label{conclusion}

We investigate the robustness of entanglement with respect to particle loss in $N$-qubit systems. We identify exhaustively all entangled states that are fragile with respect to the loss of any single qubit that belongs to any given subset of the qubits. They are given by Eq.~(\ref{canonicForm}). The entangled states that are fragile with respect to the loss of any qubits are shown to all belong to the same SLOCC class, that of the $|\mathrm{GHZ}_N\rangle$ state. In the symmetric subspace, the Majorana points on the Bloch sphere of these states are at the vertices of a regular $N$-sided polygon. This provides an interesting geometrical interpretation of the fragility of symmetric states with respect to the loss of one qubit.

We also investigate states that are fragile with respect to the loss of more than one qubit. Our results are significantly different as fragile states were identified in a continuous set of different SLOCC classes. In contrast, we finally show that all symmetric states of the $\mathcal{D}_{N-1,1}$ SLOCC class (that includes the Dicke state $|D_N^{(1)}\rangle$) are robust with respect to the loss of any number $t$ of qubits ($t \leqslant N-2$). We conjecture based on strong numerical evidence that the same holds for all symmetric states in all other Dicke state SLOCC classes $\mathcal{D}_{N-k,k}$ for $k = 2, \ldots, \lfloor N/2 \rfloor$.

\acknowledgments
A.N. acknowledges a FRIA grant and the Belgian F.R.S.-FNRS for financial support. T.B. acknowledges financial support from the Belgian F.R.S.-FNRS through IISN Grant No.\ 4.4512.08.

\providecommand{\bysame}{\leavevmode\hbox to3em{\hrulefill}\thinspace}
\providecommand{\MR}{\relax\ifhmode\unskip\space\fi MR }
\providecommand{\MRhref}[2]{%
  \href{http://www.ams.org/mathscinet-getitem?mr=#1}{#2}
}
\providecommand{\href}[2]{#2}


\begin{thebibliography}{10}
\bibitem{footnote1}
We recall that the $k$-excitation Dicke states ($k = 0, \ldots, N$) are defined as
$|D_N^{(k)}\rangle = 1/\sqrt{C_N^k} \sum_\pi |0, \dots, 0, 1, \dots, 1\rangle$,
where $C_N^k$ denotes the binomial coefficient ${N \choose k}$, the multiqubit states in the sum contain $k$ qubits in state $|1\rangle$, and $\pi$ denotes all permutations of the qubits leading to different terms in the sum. All $|D_N^{(k)}\rangle$ Dicke states ($k = 0, \ldots, N$) are symmetric and they form an orthonormal basis in the symmetric subspace of the multiqubit system. Except for $k = 0$ and $N$, they are entangled. For $k = 0, \ldots, \lfloor N/2 \rfloor$, where $\lfloor \cdot \rfloor$ denotes the floor function, the Dicke states $|D_N^{(k)}\rangle$ are SLOCC inequivalent and define in the symmetric subspace the SLOCC classes $\mathcal{D}_{N-k,k}$~\cite{Bastin_families}. In contrast, the $|D_N^{(N-k)}\rangle$ and $|D_N^{(k)}\rangle$ states are LU and thus SLOCC equivalent.
\end{thebibliography}
\end{document}